\documentclass[aps,nopacs,nokeys,superscriptaddress,11pt,twoside,notitlepage,a4paper]{revtex4-1}

\usepackage{graphicx,epic,eepic,epsfig,amsmath,latexsym,amssymb,verbatim,color,mathtools}

\usepackage{theorem}
\newtheorem{definition}{Definition}
\newtheorem{proposition}[definition]{Proposition}
\newtheorem{lemma}[definition]{Lemma}

\newtheorem{corollary}[definition]{Corollary}

\def\squareforqed{\hbox{\rlap{$\sqcap$}$\sqcup$}}
\def\qed{\ifmmode\squareforqed\else{\unskip\nobreak\hfil
\penalty50\hskip1em\null\nobreak\hfil\squareforqed
\parfillskip=0pt\finalhyphendemerits=0\endgraf}\fi}
\def\endenv{\ifmmode\;\else{\unskip\nobreak\hfil
\penalty50\hskip1em\null\nobreak\hfil\;
\parfillskip=0pt\finalhyphendemerits=0\endgraf}\fi}
\newenvironment{proof}{\noindent \textbf{{Proof~}}}{\qed}

\mathchardef\ordinarycolon\mathcode`\:
\mathcode`\:=\string"8000
\def\vcentcolon{\mathrel{\mathop\ordinarycolon}}
\begingroup \catcode`\:=\active
  \lowercase{\endgroup
  \let :\vcentcolon
  }

\newcommand{\nc}{\newcommand}
\nc{\rnc}{\renewcommand}
\nc{\beq}{\begin{equation}}
\nc{\eeq}{{\end{equation}}}
\nc{\beqa}{\begin{eqnarray}}
\nc{\eeqa}{\end{eqnarray}}
\nc{\lbar}[1]{\overline{#1}}
\nc{\bra}[1]{\langle#1|}
\nc{\ket}[1]{|#1\rangle}
\nc{\ketbra}[2]{|#1\rangle\!\langle#2|}
\nc{\braket}[2]{\langle#1|#2\rangle}
\nc{\proj}[1]{| #1\rangle\!\langle #1 |}
\nc{\avg}[1]{\langle#1\rangle}
\nc{\rank}{\operatorname{rank}}
\nc{\smfrac}[2]{\mbox{$\frac{#1}{#2}$}}
\nc{\Tr}{\operatorname{Tr}}
\nc{\ox}{\otimes}
\nc{\dg}{\dagger}
\nc{\dn}{\downarrow}
\nc{\cA}{{\cal A}}
\nc{\cB}{{\cal B}}
\nc{\cC}{{\cal C}}
\nc{\cD}{{\cal D}}
\nc{\cE}{{\cal E}}
\nc{\cF}{{\cal F}}
\nc{\cG}{{\cal G}}
\nc{\cH}{{\cal H}}
\nc{\cI}{{\cal I}}
\nc{\cJ}{{\cal J}}
\nc{\cK}{{\cal K}}
\nc{\cL}{{\cal L}}
\nc{\cM}{{\cal M}}
\nc{\cN}{{\cal N}}
\nc{\cO}{{\cal O}}
\nc{\cP}{{\cal P}}
\nc{\cQ}{{\cal Q}}
\nc{\cR}{{\cal R}}
\nc{\cS}{{\cal S}}
\nc{\cT}{{\cal T}}
\nc{\cU}{{\cal U}}
\nc{\cW}{{\cal W}}
\nc{\cX}{{\cal X}}
\nc{\cY}{{\cal Y}}
\nc{\cZ}{{\cal Z}}
\nc{\csupp}{{\operatorname{csupp}}}
\nc{\qsupp}{{\operatorname{qsupp}}}
\nc{\var}{{\operatorname{var}}}
\nc{\Var}{{\operatorname{Var}}}
\nc{\rar}{\rightarrow}
\nc{\lrar}{\longrightarrow}
\nc{\polylog}{{\operatorname{polylog}}}
\nc{\1}{{\openone}}
\nc{\wt}{{\operatorname{wt}}}
\nc{\av}[1]{{\left\langle {#1} \right\rangle}}
\newcommand{\diag}{{\operatorname{diag}}}
\newcommand{\conv}{{\operatorname{conv}}}

\nc{\RR}{{{\mathbb R}}}
\nc{\CC}{{{\mathbb C}}}
\nc{\FF}{{{\mathbb F}}}
\nc{\NN}{{{\mathbb N}}}
\nc{\ZZ}{{{\mathbb Z}}}
\nc{\PP}{{{\mathbb P}}}
\nc{\QQ}{{{\mathbb Q}}}
\nc{\UU}{{{\mathbb U}}}
\nc{\EE}{{{\mathbb E}}}
\nc{\id}{{\operatorname{id}}}

\nc{\CHSH}{{\operatorname{CHSH}}}

\nc{\be}{\begin{equation}}
\nc{\ee}{{\end{equation}}}
\nc{\bea}{\begin{eqnarray}}
\nc{\eea}{\end{eqnarray}}
\nc{\<}{\langle}
\rnc{\>}{\rangle}
\nc{\Hom}[2]{\mbox{Hom}(\CC^{#1},\CC^{#2})}
\nc{\rU}{\mbox{U}}

\nc{\ob}[1]{#1}

\nc{\SEP}{{\text{SEP}}}
\nc{\NS}{{\text{NS}}}
\nc{\LOCC}{{\text{LOCC}}}
\nc{\PPT}{{\text{PPT}}}
\nc{\EXT}{{\text{EXT}}}
\nc{\Sym}{{\operatorname{Sym}}}

\nc{\ERLO}{{E_{\text{r,LO}}}}
\nc{\ERLOCC}{{E_{\text{r,LOCC}}}}
\nc{\ERPPT}{{E_{\text{r,PPT}}}}
\nc{\ERLOCCinfty}{{E^{\infty}_{\text{r,LOCC}}}}
\nc{\Aram}{{\operatorname{\sf A}}}

\newlength{\blank}
\begin{document}

\title{Generalised Pinching Inequality}

\date{21 October 2025}

\author{Andreas Winter}
\email{andreas.winter@uni-koeln.de}
\affiliation{Department Mathematik/Informatik--Abteilung Informatik,\protect\\ Universit\"at zu K\"oln, Albertus-Magnus-Platz, 50923 K\"oln, Germany}
\affiliation{ICREA {\&} Grup d'Informaci\'o Qu\`antica, Departament de F\'isica, Universitat Aut\`onoma de Barcelona, 08193 Bellaterra (BCN), Spain}
\affiliation{Institute for Advanced Study, Technische Universit\"at M\"unchen, Lichtenbergstra{\ss}e 2a, 85748 Garching, Germany}

\begin{abstract}
Hayashi's \emph{Pinching Inequality}, which establishes a matrix inequality 
between a semidefinite matrix and a multiple of its ``pinched'' version via 
a projective measurement, has found many applications in quantum information 
theory and beyond. 
Here, we show a very simple proof of it, which lends itself immediately 
to natural generalisations where the different projections of the measurement 
have different weights, and where the matrix inequality can be reversed. 
As an application we show how the generalised pinching inequality in the case 
of binary measurements gives rise to a novel gentle measurement lemma, where 
matrix ordering replaces approximation in trace norm. 
\end{abstract}

\maketitle

\section{Pinching inequality and its generalisation}
\label{sec:general-pinch}
A projective measurement $(P_1,P_2,\ldots,P_n)$ in a 
complex Hilbert space $A$, i.e.~a collection of $n \geq 2$ projection operators 
$P_i=P_i^\dagger=P_i^2$ such that $\sum_{i=1}^n P_i = \1_A$, defines the 
\emph{pinching map} $\cP(\rho) = \sum_{i=1}^n P_i \rho P_i$, which is 
completely positive and trace preserving. 
Hayashi \cite{Hayashi:pinch} (see also \cite{HayashiOgawa:pinch}, and 
\cite{Bhatia:pinch} for other considerations on pinching), 
proved the following matrix inequality:
\begin{equation}
  \forall\, \cT(\cH) \ni \rho\geq 0 \quad \rho \leq n\,\cP(\rho),
\end{equation}
where $\rho \geq 0$ is a positive semidefinite (necessarily selfadjoint) 
trace class operator, and ``$\geq $'' refers to the L\"owner (also known as 
semidefinite) order on trace class operators. Despite its simple appearance 
and elementary proof(s), this matrix inequality has found numerous applications 
in quantum information theory and other fields. 

A slightly more abstract form was proved in \cite[Lemma~2.1]{LancienWinter:pinch}:
for arbitrary bounded operators $(M_1,M_2,\ldots,M_n)$ acting on $A$ (and 
mapping to a potentially different Hilbert space $B$),
\begin{equation}
  \forall\rho\geq 0 \quad \left(\sum_{i=1}^n M_i\right)\rho\left(\sum_{i=1}^n M_i\right)^\dagger
                            \leq n \sum_{i=1}^n M_i\rho M_i^\dagger.
\end{equation}

Here we show a flexible generalisation of the pinching inequality, and for that 
purpose define the following convex set in $\RR^n$, which is in fact 
a \emph{spectrahedron} \cite{RamanaGoldman:spectrahedron,Vinzant:spectrahedra}: 
\begin{equation}
  \label{eq:A_n}
  \cA_n := \{\vec{\alpha} = (\alpha_1,\ldots,\alpha_n)\in\RR^n : \diag(\vec{\alpha}) \geq J\},
\end{equation}
where $\diag(\vec{\alpha}) = \sum_{i=1}^n \alpha_i \proj{i}$ is the diagonal 
matrix with entries $\alpha_i$ in the computational basis; 
and $J = J_n = \sum_{i,j=1}^n \ketbra{i}{j}$ is the all-ones matrix in Dirac notation
(if it is clear from the context, we drop the -- implicit -- format $n\times n$).
%
Note that every $\vec{\alpha} \in \cA_n$ necessarily satisfies $\alpha_i > 1$ 
for all $i\in[n]$, and that $n\!\cdot\!\vec{1} = (n,n,\ldots,n) \in \cA_n$.

\begin{lemma}
  \label{lemma:general-pinch}
  If $\vec{\alpha}\in \cA_n$, then for all bounded operators $M_i:A\rightarrow B$,
  \begin{equation}
    \label{eq:general-pinch}
    \forall\rho\geq 0 \quad \left(\sum_{i=1}^n M_i\right)\rho\left(\sum_{i=1}^n M_i\right)^\dagger
                                    \leq \sum_{i=1}^n \alpha_i M_i\rho M_i^\dagger.
  \end{equation}
  Conversely, if \eqref{eq:general-pinch} holds for a nontrivial projective POVM 
  $(P_1,P_2,\ldots,P_n)$, i.e.~$P_i\neq 0$, then $\vec{\alpha}\in \cA_n$.
\end{lemma}
\begin{proof}
With $C = \CC^n$, define the bounded operator $W:A \rightarrow B\otimes C$ 
by $W := \sum_{i=1}^n M_i \ox \ket{i}$. Note that this defines a 
completely positive map $\Omega(\rho) = W\rho W^\dagger$ from 
$\cT(A)$ to $\cT(B\ox C)$, and hence for $J\leq\diag(\vec{\alpha})$,
\[
  \forall\rho\geq 0 \quad 
  \Tr_C \Omega(\rho)(\1_B\ox J) \leq \Tr_C \Omega(\rho)\left(\1_B\ox\diag(\vec{\alpha})\right). 
\]
Clearly, the left hand side of the inequality equals 
$\left(\sum_{i=1}^n M_i\right)\rho\left(\sum_{i=1}^n M_i\right)^\dagger$, 
while the right hand side is 
$\sum_{i=1}^n \alpha_i M_i\rho M_i^\dagger$, proving the first part of the lemma. 

Conversely, if for a projective POVM $(P_1,\ldots,P_n)$ and for all $\rho\geq 0$, 
\[
  \rho \leq \sum_{i=1}^n \alpha_i P_i\rho P_i,
\]
choose fixed point unit vectors $\ket{e_i}$ of $P_i$, i.e.~$P_i\ket{e_i}=\ket{e_i}$, 
and consider $\rho = \sum_{i,j=1}^n \ketbra{e_i}{e_j}$. Then the above inequality 
reduces to 
\[
  \sum_{i,j=1}^n \ketbra{e_i}{e_j} \leq \sum_{i=1}^n \alpha_i \proj{e_i},
\]
which clearly is equivalent to $\diag(\vec{\alpha}) \in \cA_n$.
\end{proof}

\medskip
Membership in $\cA_n$ is a simple semidefinite constraint, and by inspection 
one sees that it is given by the polynomial criterion of 
$\det\left(\diag(\vec{\alpha})-J\right) \geq 0$ and all nontrivial 
minors (the determinants of $k\times k$-principal submatrices for $k<n$) 
being positive. By Sylvester's condition, any nested family of 
$k\times k$-minors for $k=1,\ldots,n-1$ is sufficient \cite{HornJohnson}. 
Still, due to the occurrence of the determinant, for larger $n$ the 
formulas rapidly become quite cumbersome. We can however construct the 
rational polynomial conditions for $\vec{\alpha} \in \cA_n$ recursively 
starting from $n=2$ and using the Schur complement \cite{Haynsworth,Zhang:Schur}.
Indeed, for $n=2$ we have $(\alpha_1,\alpha_2)\in A_2$ iff 
$\alpha_1 > 1$ and $(\alpha_1-1)(\alpha_2-1)\geq 1$. 
For $n>2$, the Schur complement tells us that 
$\diag(\alpha_1,\alpha_2,\ldots,\alpha_n)-J_n \geq 0$ iff 
\[
 \left[\diag(\alpha_1,\ldots,\alpha_{n-1})-J_{n-1}\right] 
   \oplus \left[\alpha_n-1-\Tr J_{n-1}\left(\diag(\alpha_1,\ldots,\alpha_{n-1})-J_{n-1}\right)^{-1}\right] \geq 0.
\]
The positive semidefiniteness of $\diag(\alpha_1,\alpha_2,\ldots,\alpha_n)-J_n$ 
is thus equivalent to
\begin{align}
  \diag(\alpha_1,\ldots,\alpha_{n-1}) 
                       &> J_{n-1}, \\ 
  \text{and } \alpha_n &\geq \Tr\left(\diag(\alpha_1,\ldots,\alpha_{n-1})
                                      \bigl(\diag(\alpha_1,\ldots,\alpha_{n-1})-J_{n-1}\bigr)^{-1}\right).
\end{align}
Note that the first condition says $(\alpha_1,\ldots,\alpha_{n-1}) \in \cA_{n-1}^\circ$, 
the interior of $\cA_{n-1}$.

As an example, we get $(\alpha_1,\alpha_2,\alpha_3)\in\cA_3$ if and only if
$\alpha_1 > 1$, $\alpha_2 > 1 + 1/(\alpha_1-1)$ and 
$\left[(\alpha_1-1)(\alpha_2-1) - 1\right]\left[(\alpha_1-1)(\alpha_3-1) - 1\right] \geq \alpha_1^2$. 
For $n=4$, the same three inequalities emerge again (with the last one turning from ``$\geq$''
to strict inequality), plus a fourth one involving $\alpha_4$, etc. 
We will return to the case $n=2$ in Section \ref{sec:binary-pinch} below.

\section{Reverse pinching inequality}
\label{sec:reverse-pinch}
The proof of Lemma \ref{lemma:general-pinch} suggests a direct way to 
obtaining lower bounds on $\rho$, or more generally on 
$\left(\sum_{i=1}^n M_i\right)\rho\left(\sum_{i=1}^n M_i\right)^\dagger$, namely 
to consider $\diag(\vec{\beta}) \leq J$ for real vectors $\vec{\beta}\in\RR^n$,
forming an ``opposite'' spectrahedron $\cB_n$ to $\cA_n$. 
Using the same proof, we obtain the following:
\begin{lemma}
  \label{lemma:reverse-pinch}
  If $\vec{\beta}\in \cB_n$, then for all bounded operators $M_i:A\rightarrow B$,
  \begin{equation}
    \label{eq:reverse-pinch}
    \phantom{===}
    \forall\rho\geq 0 \quad \left(\sum_{i=1}^n M_i\right)\rho\left(\sum_{i=1}^n M_i\right)^\dagger
                                    \geq \sum_{i=1}^n \beta_i M_i\rho M_i^\dagger.
    \phantom{===}\qed
  \end{equation}
\end{lemma}

For $\vec{\beta}\in \cB_n$, it is necessary that all $\beta_i < 1$, but as $J$ 
has a single eigenvalue $n$ and all other eigenvalues $0$, either all $\beta_i \leq 0$, 
or else exactly one $\beta_i > 0$ and all other $\beta_j < 0$ ($j\neq i$). 
This suggests that the reverse inequality \eqref{eq:reverse-pinch} is 
most interesting in the case $n=2$, to which we will turn next.

\section{Binary case: gentle measurement inequality}
\label{sec:binary-pinch}
For $n=2$, we are facing $(\alpha_1,\alpha_2)\in \cA_2$ iff 
\[
  \left[\begin{matrix}
    \alpha_1-1 & -1 \\
            -1 & \alpha_2-1
  \end{matrix}\right] \geq 0,
\]
which is equivalent to $\alpha_1,\,\alpha_2 > 1$ and $(\alpha_1-1)(\alpha_2-1)\geq 1$. 

Similarly, $(\beta_1,\beta_2)\in \cB_2$ iff 
\[
  \left[\begin{matrix}
          1-\beta_1 & 1 \\
                  1 & 1-\beta_2
        \end{matrix}\right] \geq 0,
\]
which boils down to $\beta_1,\,\beta_2 < 1$ and $(1-\beta_1)(1-\beta_2)\geq 1$. 

Of course, the tightest bound is obtained when we have equality in the determinant 
condition, proving the following:
\begin{proposition}
  \label{prop:A_2:B_2}
  We have
  \begin{align*}
    \cA_2 &= \conv\left\{\left(1+t,1+\frac1t\right) : 0<t<\infty \right\}, \\
    \cB_2 &= \conv\left\{\left(1-t,1-\frac1t\right) : 0<t<\infty \right\}.
  \end{align*}
  Thus, for $0<t\leq 1$ and any two bounded operators $M_1$ and $M_2$,
  \begin{equation}
    \label{eq:rho-M1-M2}
    (1-t)M_1\rho M_1^\dagger - \left(\frac1t - 1\right)M_2 \rho M_2^\dagger
       \leq (M_1+M_2)\rho(M_1+M_2)^\dagger 
       \leq (1+t)M_1\rho M_1^\dagger  + \left(1+\frac1t\right)M_2 \rho M_2^\dagger.
  \end{equation}
\end{proposition}
\begin{proof}
We have argued the form of $\cA_2$ and $\cB_2$ already, which we only have 
to plug into Lemmas \ref{lemma:general-pinch} and \ref{lemma:reverse-pinch}
to obtain \eqref{eq:rho-M1-M2}.
However, there is also an entertaining direct proof: for the right hand 
inequality, consider 
\[\begin{split}
  (M_1+M_2)\rho(M_1+M_2)^\dagger
    &\leq (M_1+M_2)\rho(M_1+M_2)^\dagger 
          + \left(t^{\frac12}M_1-t^{-\frac12}M_2\right)\rho\left(t^{\frac12}M_1-t^{-\frac12}M_2\right)^\dagger \\
    &=    (1+t)M_1\rho M_1^\dagger  + \left(1+\frac1t\right)M_2 \rho M_2^\dagger, 
\end{split}\]
by simply distributively expanding the products and cancellation of mixed terms. 
Likewise, 
\[\begin{split}
  (M_1+M_2)\rho(M_1+M_2)^\dagger
    &\geq (M_1+M_2)\rho(M_1+M_2)^\dagger 
          - \left(t^{\frac12}M_1+t^{-\frac12}M_2\right)\rho\left(t^{\frac12}M_1+t^{-\frac12}M_2\right)^\dagger \\
    &=    (1-t)M_1\rho M_1^\dagger  + \left(1-\frac1t\right)M_2 \rho M_2^\dagger,
\end{split}\]
and we are done. 
\end{proof}


\begin{corollary}
  \label{cor:gentle-measurement}  
  If $\Tr\rho P \geq 1-\epsilon$, then
  \begin{align*}
    \rho &\leq \left(1-\sqrt{\epsilon}\right)P\rho P 
                + \left(1+\frac{1}{\sqrt{\epsilon}}\right) P^\perp\rho P^\perp, \\
    \rho &\geq \left(1-\sqrt{\epsilon}\right)P\rho P 
                + \left(1-\frac{1}{\sqrt{\epsilon}}\right)P^\perp\rho P^\perp.
  \end{align*}
  In particular, with the subnormalised state $\rho^\perp = \frac{1}{\epsilon}P^\perp\rho P^\perp$, 
  \begin{equation}
    -\sqrt{\epsilon}P\rho P - \sqrt{\epsilon} \rho^\perp
      \leq \rho - P\rho P 
      \leq \sqrt{\epsilon}P\rho P + \left(\epsilon+\sqrt{\epsilon}\right) \rho^\perp.
  \end{equation}
\end{corollary}

\begin{proof}
We simply choose $t=\sqrt{\epsilon}$ in Lemmas \ref{lemma:general-pinch}
and \ref{lemma:reverse-pinch} (or directly in Proposition \ref{prop:A_2:B_2}).
\end{proof}

\medskip
Note that by taking the trace norm in the last inequality, we obtain that 
$\Tr\rho P \geq 1-\epsilon$ implies $\frac12\|\rho-P\rho P\|_1 \leq \sqrt{\epsilon}+\epsilon$, 
a ``gentle measurement lemma'' improving on the original \cite{AW:gentle-measurement}
(which had the suboptimal upper bound $2\sqrt{\epsilon}$), but falling short of the 
improved bound $\sqrt{\epsilon}$ in \cite{OgawaNagaoka:gentle-measurement} and the 
further tightened bounds in \cite[Lemma~6]{RLD:relative-entropies}. On the other hand, 
the matrix inequalities in Corollary \ref{cor:gentle-measurement} give more 
information, which might be useful in certain circumstances. It would 
be interesting to find an application of the present ordered gentle measurement lemma, 
or for that matter, of the generalised and reverse pinching inequalities.

\section*{Acknowledgments}
The author thanks Matt Hoogsteder-Riera, John Calsamiglia and Frank Vallentin 
for useful discussions on pinching inequalities, 
and Lenina Huxley and Alfredo Garcia for indirectly suggesting the connection 
with the gentle measurement lemma. 
He furthermore acknowledges support by the European Commission QuantERA project
ExTRaQT (Spanish MICIN grant no.~PCI2022-132965), by the Spanish MICIN 
(project PID2022-141283NB-I00) with the support of FEDER funds, 
by the Spanish MICIN with funding from European Union NextGenerationEU 
(PRTR-C17.I1) and the Generalitat de Catalunya, by the Spanish MTDFP 
through the QUANTUM ENIA project: Quantum Spain, funded by the European 
Union NextGenerationEU within the framework of the ``Digital Spain 
2026 Agenda'', by the Alexander von Humboldt Foundation, and by the 
Institute for Advanced Study of the Technical University Munich.


\begin{thebibliography}{10}
\bibitem{Hayashi:pinch} Masahito Hayashi, 
  ``Optimal sequence of POVMs in the sense of Stein's lemma in quantum hypothesis testing'',
  J. Phys. A: Math. Gen. {\bf 35}(5):10759-10773, 2002;
  doi:10.1088/0305-4470/35/50/307;
  arXiv:quant-ph/0107004.

\bibitem{HayashiOgawa:pinch} Tomohiro Ogawa and Masahito Hayashi, 
  ``On Error Exponents in Quantum Hypothesis Testing'', 
  IEEE Trans. Inf. Theory {\bf 50}(6):1368-1372, 2004;
  doi:10.1109/TIT.2004.828155;
  arXiv:quant-ph/0206151.

\bibitem{Bhatia:pinch} Rajendra Bhatia, 
  ``Pinching, Trimming, Truncating, and Averaging of Matrices'', 
  Amer. Math. Monthly {\bf 107}(7):602-608, 2000;
  doi:10.2307/2589115. 

\bibitem{LancienWinter:pinch} Cecilia Lancien and Andreas Winter, 
  ``Flexible Constrained de Finetti Reductions and Applications'',
  J. Math. Phys. {\bf 58}:092203, 2017; 
  doi:10.1063/1.5003633;
  arXiv[quant-ph]:1605.09013.

\bibitem{RamanaGoldman:spectrahedron} Motakuri Ramana and Alan J. Goldman, 
  ``Some Geometric Results in Semidefinite Programming'',
  J. Global Optimization {\bf 7}(1):33-50, 1995;
  doi:10.1007/BF01100204.

\bibitem{Vinzant:spectrahedra} Cynthia Vinzant, 
  ``The geometry of spectrahedra'', 
  in: \emph{Sum of Squares: Theory and Applications. AMS Short Course, 14-15 Jan. 2019, Baltimore MD} 
  (P. A. Parrilo and R. R. Thomas, eds.), 
  Proc. Symposia Applied Mathematics, vol. 77, pp.~11-35, AMS, 2020; 
  doi:10.1090/psapm/077. 

\bibitem{HornJohnson} Roger A. Horn and Charles R Johnson, 
  \emph{Matrix Analysis}, Cambridge University Press, Cambridge, 1985; 
  doi:10.1017/CBO9780511810817.

\bibitem{Haynsworth} Emilie V. Haynsworth, 
  ``Determination of the Inertia of a Partitioned Hermitian Matrix'', 
  Lin. Algebra Appl. {\bf 1}(1):73-81, 1968;
  doi:10.1016/0024-3795(68)90050-5.

\bibitem{Zhang:Schur} Fuzhen Zhang (ed.), \emph{The Schur Complement and Its Applications},
  Numerical Methods and Algorithms, Springer New York, 2005.

\bibitem{AW:gentle-measurement} Andreas Winter, 
  ``Coding Theorem and Strong Converse for Quantum Channels'', 
  IEEE Trans. Inf. Theory {\bf 45}(7):2481-2485, 1999;
  doi:10.1109/18.796385;
  arXiv[quant-ph]:1409.2536. 
 
\bibitem{OgawaNagaoka:gentle-measurement} Tomohiro Ogawa and Hiroshi Nagaoka, 
  ``Making Good Codes for Classical-Quantum Channel Coding via Quantum Hypothesis Testing'',
  IEEE Trans. Inf. Theory {\bf 53}(6):2261-2266, 2007;
  doi:10.1109/TIT.2007.896874;
  arXiv:quant-ph/0208139. 

\bibitem{RLD:relative-entropies} Bartosz Regula, Ludovico Lami, and Nilanjana Datta, 
  ``Tight relations and equivalences between smooth relative entropies'',
  arXiv[quant-ph]:2501.12447, 2025.

\end{thebibliography}
\end{document}